\documentclass{llncs}
\usepackage{amsmath}
\usepackage{amssymb}
\usepackage{enumerate}
\usepackage{varioref}
\usepackage{charter,eulervm}

\title{{\sc Edge-Clique Graphs of Cocktail Parties have 
Unbounded Rankwidth}}

\author{
 Maw-Shang~Chang\inst{1}
\and 
 Ton~Kloks\inst{2} 
\and
 Ching-Hao~Liu\inst{2}} 
\institute{
 Department of Computer Science and Information Engineering\\
 Hungkuang University, Taiwan\\
 {\tt mschang@sunrise.hk.edu.tw}
\and 
 Department of Computer Science\\
 National Tsing Hua University, Taiwan\\
 {\tt chinghao.liu@gmail.com}
}

\begin{document}

\maketitle

\begin{abstract}
In an attempt to find a polynomial-time 
algorithm for the edge-clique cover problem on cographs   
we tried to prove that the edge-clique graphs of cographs have bounded 
rankwidth.  
However, this is not the case. 
In this note 
we show that the edge-clique graphs of cocktail party graphs 
have unbounded rankwidth.
\end{abstract}

\section{Introduction}

Let $G=(V,E)$ be an undirected graph with vertex set $V$ and 
edge set $E$. A clique is a complete subgraph of $G$. 

\begin{definition}
A edge-clique covering of $G$ is a family of 
complete subgraphs such that 
each edge of $G$ is in at least one member of the family. 
\end{definition}
The minimal cardinality of such a family is the 
edge-clique covering number, and we denote it by 
$\theta_e(G)$. 

\bigskip 

The problem of deciding if $\theta_e(G) \leq k$, 
for a given natural number $k$, 
is NP-complete~\cite{kn:kou,kn:orlin,kn:holyer}. 
The problem remains NP-complete when restricted to graphs 
with maximum degree at most six~\cite{kn:hoover}.  
Hoover~\cite{kn:hoover} gives a polynomial time algorithm for 
graphs with maximum degree at most five. 
For graphs with 
maximum degree less than five, this was already done 
by Pullman~\cite{kn:pullman}. 
Also for linegraphs the problem can be solved in 
polynomial time~\cite{kn:orlin,kn:pullman}. 
 
In~\cite{kn:kou} it is shown that approximating the 
clique covering number within a constant factor smaller than two remains 
NP-complete. 

\bigskip 

Gy\'arf\'as~\cite{kn:gyarfas} showed the following 
interesting lowerbound. 
Two vertices $x$ and $y$ are {\em equivalent\/} if 
they are adjacent and have the same closed neighborhood. 

\begin{theorem}
\label{gyarfas}
If a graph $G$ has $n$ vertices and contains neither isolated 
nor equivalent vertices then $\theta_e(G) \geq \log_2(n+1)$. 
\end{theorem}

Gy\'arf\'as result implies that the edge-clique 
cover problem is fixed-parameter 
tractable (see also~\cite{kn:gramm}). 
Cygan et al showed that, under the assumption of the 
exponential time hypothesis, there is 
no polynomial-time algorithm which reduces the 
parameterized problem $(\theta_e(G),k)$ to a kernel of size bounded 
by $2^{o(k)}$. In their proof the authors make use of the fact that 
$\theta_e(cp(2^{\ell}))$ is a [sic] 
``hard instance for the edge-clique cover problem, at least from a 
point of view of the currently known algorithms.''    
Note that, in contrast, the parameterized edge-clique partition 
problem can be reduced to a kernel with at most $k^2$ 
vertices~\cite{kn:mujuni}. (Mujuni and Rosamond also mention that 
the edge-clique cover problem probably has no polynomial kernel.)   

\section{Rankwidth of edge-clique graphs of cocktail parties}

\begin{definition}
The cocktail party graph $cp(n)$ is the complement of a 
matching with $2n$ vertices. 
\end{definition}

Notice that a cocktail party graph has no 
equivalent vertices. Thus, by Theorem~\ref{gyarfas}, 
\[\theta_e(cp(n)) \geq log_2(2n+1).\] 
For the cocktail party graph 
an exact formula for $\theta_e(cp(n))$ 
is given in~\cite{kn:gregory}. 
In that paper Gregory and Pullman prove that 
\[\lim_{n \rightarrow \infty} \; \frac{\theta_e(cp(n))}{\log_2(n)}=1.\] 

\bigskip 
   
\begin{definition}
Let $G=(V,E)$ be a graph. The edge-clique graph $K_e(G)$ has 
as its vertices the edges of $G$ and two vertices of 
$K_e(G)$ are adjacent when the corresponding edges in $G$ are 
contained in a clique. 
\end{definition}

Albertson and Collins prove that there is a 1-1 correspondence 
between the maximal cliques in $G$ and $K_e(G)$~\cite{kn:albertson}. 
The same holds true 
for the intersections of maximal cliques in $G$ and in $K_e(G)$. 

\bigskip 

For a graph $G$ we denote the vertex-clique cover number of 
$G$ by $\kappa(G)$. 
Thus 
\[\kappa(G) = \chi(\Bar{G}).\]  
Notice that, for a graph $G$,  
\[\theta_e(G)=\kappa(K_e(G)).\] 

\bigskip 

Albertson and Collins mention the following result (due to 
Shearer)~\cite{kn:albertson} for the graphs $K_e^r(cp(n))$, 
defined inductively by $K_e^r(cp(n))=K_e(K_e^{r-1}(cp(n)))$.
\[\alpha(K_e^r(cp(n))) \leq 3\cdot (2^r)!\] 
Thus, for $r=1$, $\alpha(K_e(cp(n))) \leq 6$. 
However, the following is easily checked. 

\begin{lemma}
\label{bound alpha}
For $n \geq 2$ 
\[\alpha(K_e(cp(n))) =4.\]
\end{lemma}
\begin{proof}
Let $G$ be the complement of 
a matching $\{x_i,y_i\}$, for $i \in \{1,\dots,n\}$. 
Let $K=K_e(G)$. 
Obviously, every pair of edges in the matching induces 
an independent set with four vertices in $K$. 

\medskip 

\noindent
Consider an edge $e=\{x_i,x_j\}$ in $G$. The only 
edges in $G$ that are not adjacent to $e$ in $K$, must have 
endpoints in $y_i$ or in $y_j$. Consider an edge 
$f=\{y_i,y_k\}$ for some $k \notin \{i,j\}$. The only other 
edge incident with $y_i$, 
which is not adjacent in $K$ to $f$ nor to $e$ is 
$\{y_i,x_k\}$. 

\medskip 

\noindent
The only edge incident with $y_j$ which is not adjacent to 
$e$ nor $f$ is $\{y_j,x_i\}$. This proves the lemma. 
\qed\end{proof}

\bigskip 

\begin{definition}
A class of graphs $\mathcal{G}$ is $\chi$-bounded if there 
exists a function $f$ such that for every graph $G \in \mathcal{G}$, 
\[\chi(G) \leq f(\omega(G)).\]
\end{definition}
Dvo\v{r}\'ak and Kr\'al proved that the class of graphs 
with rankwidth at most $k$ is $\chi$-bounded~\cite{kn:dvorak}. 

\bigskip 

We now easily obtain our result. 

\begin{theorem}
The class of edge-clique graphs of cocktail parties has 
unbounded rankwidth. 
\end{theorem}
\begin{proof}
It is easy to see that the rankwidth of any graph 
is at most one more than the rankwidth of its complement~\cite{kn:oum}. 
Assume that there is a constant $k$ such that 
the rankwidth of $K_e(G)$ is at most $k$ whenever $G$ is a 
cocktail party graph. Let 
\[\mathcal{K}=\{\; \overline{K_e(G)} \;|\; G \simeq cp(n), 
\;n \in \mathbb{N}\;\}.\] 
Then the rankwidth of graphs in $\mathcal{K}$ is uniformly bounded 
by $k+1$. By the result of Dvo\v{r}\'ak and Kr\'al, there exists 
a function $f$ such that 
\[\kappa(K_e(G)) \leq f(\alpha(K_e(G)))\] 
for every cocktail party graph $G$. This contradicts 
Lemma~\ref{bound alpha} and Theorem~\ref{gyarfas}. 
\qed\end{proof}
  
\section{Concluding remark}

As far as we know, the recognition of 
edge-clique graphs is an open problem. 

\begin{conjecture}
The edge-clique cover problem is NP-complete for cographs. 
\end{conjecture}

\end{document}